\documentclass[a4paper]{article}
\usepackage{a4wide}
 \usepackage[utf8]{inputenc}
\usepackage{microtype}%if unwanted, comment out or use option "draft"
\usepackage[ruled,vlined]{algorithm2e}
\usepackage{verbatim}
\usepackage{amssymb}
\usepackage{amsthm}
\usepackage{amsxtra}
\usepackage{amsmath}
\usepackage{amscd}
\usepackage{epstopdf}
\usepackage{graphicx}
\newtheorem{theorem}{Theorem}[section]
\newtheorem{lemma}[theorem]{Lemma}

\newtheorem{proposition}[theorem]{Proposition}
\newtheorem{corollary}[theorem]{Corollary}

\newtheorem{definition}[theorem]{Definition}

\newtheorem{example}{Example}[section]

\newtheorem{remark}[theorem]{Remark}

\begin{document}

\title{A General Method for Common Intervals}

\author{Ismael Belghiti \thanks{École Normale Supérieure de Paris, France. Email : \texttt{ismael.alaoui.belghiti@ens.fr}} ~~~and M. Habib \thanks{LIAFA,  CNRS  \&  Universit\'e Paris Diderot - Paris 7, France. Email: \texttt{habib@liafa.univ-paris-diderot.fr}
}}

\maketitle

\begin{abstract}

Given an elementary chain of vertex set $V$, 
seen as a labelling of $V$ by the set ${\{1,\ldots,n=|V|\}}$, 
and another discrete structure over $V$, say a graph $G$, 
the \textbf{problem of common intervals} is to compute the induced subgraphs $G[I]$, 
such that $I$ is an interval of $[1, n]$ and $G[I]$ 
satisfies some property $\Pi$ (as for example $\Pi=$ "being connected").
This kind of problems comes from comparative genomic in bioinformatics, 
mainly when the graph $G$ is  a chain or a tree \cite{heb_stoye, interval_tree, BBCR04}.

When the family of intervals is closed under intersection, we present here the  combination of  two approaches, 
namely the idea of potential beginning developed in \cite{uno_yag, revisit_uno} and the notion of generator as defined in \cite{BCMR08}. This yields
 a very simple generic algorithm to compute all common intervals, which gives optimal algorithms in various applications. 
For example in the case where $G$ is a tree, our framework yields the first linear time algorithms for the 
two properties: "being connected" and "being a path". In the case where $G$ is a chain, the problem is known as: 
\textbf{common intervals of two permutations} \cite{uno_yag}, our algorithm provides not only the set of all common 
intervals but also with some easy modifications a tree structure that represents this set.

\end{abstract}

\section{Introduction}

All the graphs considered here are supposed to be finite, undirected, simple and  loopless.  
For such a graph $G$, we denote by $V(G)$ and $E(G)$ its vertex and edge sets respectively. Furthermore if $U$ is a subset of $V(G)$, we denote by $G[U]$
the induced subgraph of $G$.

The problem of finding the \textbf{common connected components} of two graphs was defined in \cite{BMLPV05}, as follows:
let $G_1$, $G_2$ be two graphs on the same vertex set $V$, find the maximal partition $\mathcal Q=\{V_1, \dots, V_k\}$ of $V$ such for that for every $i \in [1, k]$, $G_1[V_i]$ and $G_2[V_i]$ are connected. 
Of course this problem is polynomially tractable and some subcases are solvable in linear time (see \cite{BXHP08}).

In this paper we are particularly interested in the close problem of finding all the \textbf{common connected subsets} of two graphs:
let $G_1$, $G_2$ be two graphs on the same vertex set $V$, find all the subsets $U \subset V$ such that 
$G_1[U]$ and $G_2[U]$ are connected.

%We define the \textit{connected intervals} of $G$ as the integer intervals $I \subset V$ such that $G[I]$ is connected and
%we denote $\mathcal{I}$ this set of intervals. In this paper, we address the problem of computing $\mathcal{I}$ when $G$ is a tree or,
%more specifically, a path.

More precisely we mainly study the particular case where $G_1$ is a elementary chain, 
seen as a labelling of $V$ by $\{1,\ldots,n=|V|\}$,
and $G_2$ is a graph $G$ with  $V(G)=\{1, \dots n \}$, the previous problem becomes the \textbf{problem of common intervals}. That is to compute the induced subgraphs $G[I]$, such that $I$ is an interval of $[1, n]$ and $G[I]$ satisfies some property $\Pi$ (as for example $\Pi=$ "being connected").
This kind of problems appears in Biology from comparative genomic, in a more specific case when the graph $G$ is  a chain or a tree \cite{heb_stoye, interval_tree, BBCR04}.

Combining two approaches namely the idea of potential beginning developed in \cite{uno_yag, revisit_uno} and the notion of generator as defined in \cite{BCMR08},
we succeed to a obtain a very simple generic algorithm which yields optimal algorithms in various applications. For example in the case where $G$ is a tree, our framework yields the first linear time algorithms for the two properties: "being connected" and "being a path".

Furthermore in the particular  case where $G$ is a chain, we deal with common intervals of two permutations, although some linear time algorithms already exist \cite{uno_yag, revisit_uno, BCMR08}, our framework yields very simple linear time algorithms that compute non only the common intervals, but also the associated tree decomposition.

In this paper we will first present the general framework, which deals with  families of intervals closed by intersection and then describe how the generic algorithms can be specialized for some applications. Due to space constraints we will not develop in details all these applications.

\section{General Framework for families closed under intersection}

In the sequel, we only consider families of intervals closed under intersection. 
In other words, we will consider families $\mathcal{F}$ of intervals such that:
if two intervals $I_1, I_2 \in \mathcal{F}$ intersect then their intersection
$I_1 \cap I_2$ is also in $\mathcal{F}$. For example, in the cases where $G$ is a tree and $\Pi=$ "being connected" or 
$\Pi=$ "being a path", the resulting families are closed by intersection. But it is not always true, as for the case
where $G$ is a graph and $\Pi=$ "being connected", and to manage this case we need to extend the framework presented here, see \cite{BH13}.

In the whole section, we assume by convention that the considered
families of intervals contain all the singletons of their ground set.
Let us now describe a generic algorithm to compute 
a convenient representation for these families and another one to enumerate their elements.
These algorithms will be specialized different ways in the section \ref{applications},
according to the particular combinatorial structures we consider.

\subsection{Representation by a generator}

\cite{BCMR08} introduced the notion of \textit{generator} to 
represent in linear space families of intervals closed under intersection:

\begin{definition}[Generator]
A generator of a family $\mathcal{F}$ of intervals over $\{1,\ldots, n\}$ closed under intersection is a couple
of vectors $(L,R)$ such that:
\begin{itemize}
   \item $\forall x \in \{1,\ldots,n\}, R[x] \geq x$
   \item $\forall y \in \{1,\ldots,n\}, L[y] \leq y$
   \item $[x,y] \in \mathcal{F} \Longleftrightarrow R[x] \geq y$ and  $L[y] \leq x$ 
\end{itemize}
\end{definition}

The following lemma shows that the families of intervals closed under intersection do admit such
a representation.

\begin{lemma}[\cite{BCMR08} Existence of a representation by generator]
Let $\mathcal{F}$ be a family of intervals closed under intersection.
There exists a generator that represents the family $\mathcal{F}$.
\begin{proof}
Let $maxEnd[x]$ be the maximum end of an interval of $\mathcal{F}$ starting at $x$ 
and let $minBeg[y]$ be the minimum beginning of an interval of $\mathcal{F}$ ending at $y$.
$(maxEnd, minBeg)$ is a generator of $\mathcal{F}$.
\end{proof}
\end{lemma}

Notice that this representation is particularly useful when we want to consider
the intersection of several families of intervals closed under intersection:

\begin{lemma}[\cite{BCMR08} Generators and intersection of families]
If $\mathcal{F}_1$ and $\mathcal{F}_2$ are two families of intervals closed under intersection, 
$\mathcal{F}_1$ being represented by the generator $(L_1,R_1)$ and $\mathcal{F}_2$ by the generator
$(L_2,R_2)$, then $\mathcal{F}_1 \cap \mathcal{F}_2$ is represented by the generator $(L,R)$ defined by:
\begin{enumerate}
\item $\forall x \in \{1,\ldots,n\}, R[x] = min(R_1[x],R_2[x])$
\item $\forall y \in \{1,\ldots,n\}, L[y] = max(L_1[y],L_2[y])$
\end{enumerate} 
\end{lemma}

\begin{figure}[!h]
\begin{center}
\begin{minipage}[c]{0.4\textwidth}
\includegraphics[scale=0.6]{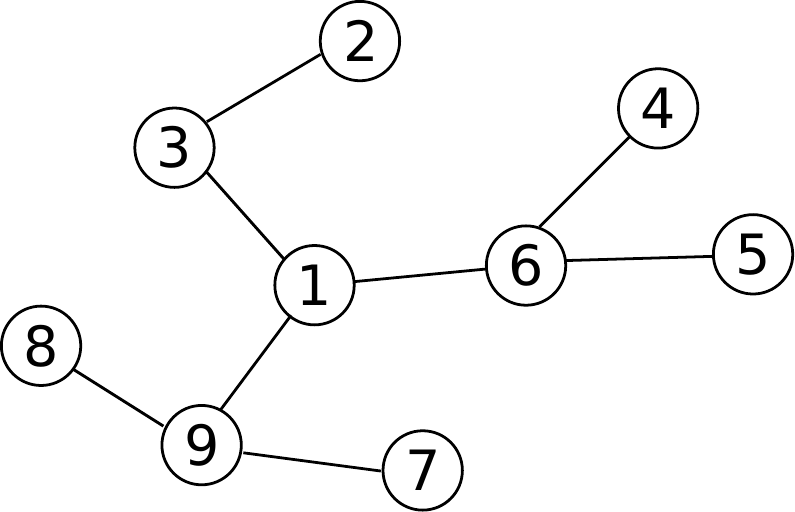}
\end{minipage}
\begin{minipage}[c]{0.4\textwidth}
   \begin{tabular}{| c | c | c | c | c | c | c | c | c | c | }
      \hline
      \empty & 1 & 2 & 3 & 4 & 5 & 6 & 7 & 8 & 9 \\
      \hline
      R & 9 & 4 & 4 & 7 & 7 & 7 & 9 & 9 & 9 \\ 
      \hline
      L & 1 & 1 & 1 & 3 & 4 & 1 & 6 & 7 & 1 \\
      \hline
   \end{tabular}
\end{minipage}
\caption{A labelled tree $T$ with, on its right, an example of generator representing the family of the intervals
$I$ such that $T[I]$ is connected.
For example, we have $R[1] = 9$ and $L[6] = 1$ thus, since $R[1] \geq 6$ and $L[6] \leq 1$, $[1,6]$ is
in this family.}
\end{center}
\end{figure}

\subsection{\label{potBeg}Potential Beginnings and Right-Splitters}

Let $\mathcal{F}$ be a family of intervals over $\{1,\ldots,n\}$ closed under intersection.
An element of $\{1,\ldots,n\}$ will be called a \textit{vertex}. 
In our algorithms, we will consider the vertices $y$ in increasing order and 
we will be interested in the beginnings of the intervals of $\mathcal{F}$ 
ending at $y$. When we have considered the vertices $\{1, \ldots, y\}$, 
it is natural to only keep the vertices $x$ that could be the beginning of an interval of 
$\mathcal{F}$ ending after $y$. 
To capture this idea let us introduce a notion
of \textit{potential beginning} (with respect to $y$), that relies on a simple property
such $x$ have to satisfy, and symmetrically a notion of \textit{potential end} such that:
$[x,y]$ is in $\mathcal{F}$ iff $x$ is a potential beginning of $y$
and $y$ is a potential end of $x$.

\begin{example}
Assume that we are given a permutation $P$ over $\{1,\ldots,n\}$
and that we are interested in the family $\mathcal{F}$ of intervals 
$[x,y] \subset \{1,\ldots,n\}$ such that $P([x,y]) = [x,y]$
(remark that $\mathcal{F}$ is closed under intersection).
Defining the potential beginnings of $y$ as the $x \leq y$ such 
that $x \leq \min P([x,y])$ and the potentials ends of $x$ as the $y \geq x$
such that $y \geq \max P([x,y])$, it is straightforward to check that: 
$[x,y]$ is in $\mathcal{F}$ iff $x$ is a potential beginning of $y$
and $y$ is a potential end of $x$.
\end{example}

We thus introduce the following definition:

\begin{definition}
A \textit{couple of potentiality} for the family $\mathcal{F}$ is a couple $(potBeg,potEnd)$ such that:
\begin{itemize}
	\item $\forall x,y\in \{1,\ldots,n\}, potBeg(y) \subset \{1,\ldots,y\}$ and $potEnd(x) \subset \{x,\ldots,n\}$
	\item $\forall y \in \{2,\ldots,n\}$,   $potBeg(y) \subset potBeg(y-1)\cup \{y\}$.
	\item $\forall x \in \{1,\ldots,n-1\}$, $potEnd(x) \subset potEnd(x+1)\cup \{x\}$.
	\item $\forall 1 \leq x \leq y \leq n, [x,y]\in\mathcal{F} \Leftrightarrow (x\in potBeg(y)$ and $ y\in potEnd(x))$.
\end{itemize} 

The elements of $potBeg(y)$ are called the \textit{potential beginnings} of $y$
and the elements of $potEnd(x)$ are called the \textit{potential ends} of $x$.
\end{definition}

Remark that we want the notion of \textit{potential beginning} to be such that,
when we consider the vertices $y$ in increasing order, each vertex $x$ will be 
a \textit{potential beginning} during a certain time until it loses its \textit{potentiality}.

\begin{definition}
We define the \textit{right-splitter} of $x$, denoted $RSplitter[x]$, 
as the minimum $y>x$ such that $x \not\in potBeg(y)$ (if such an index $y$ does not exist,
we set $RSplitter[x]=\infty$). Symmetrically, we define the \textit{left-splitter} 
of $y$, denoted $LSplitter[y]$, as the maximum $x<y$ such that $y \not\in potEnd(x)$ 
(if such an $x$ does not exist, we set $LSplitter[y]=-\infty$)
\end{definition}

From the above definitions, we have a straightforward link between 
the notion of potential beginning and right-splitter:

\begin{proposition} [Link between potential beginnings and right-splitters]
For $x,y\in \{1,\ldots,n\}$ with $x \leq y$, 
$x$ is a potential beginning of $y$ iff $y < RSplitter[x]$.
Symmetrically: $y$ is a potential end of $x$ iff $x > LSplitter[y]$.

Hence, we have that $(LSplitter+1,RSplitter-1)$ is a generator of $\mathcal{F}$.
\end{proposition}

Left and right splitters and potential extremities are somehow dual and give two different points of view of the same structural
properties.
During a sweep taking the vertices $y$ in increasing order, we will maintain the set \texttt{PotBeg}
of the potential beginnings of $y$ in a data structure depending on the application.

\begin{remark} 
Each vertex will be pushed once in this structure and thus removed at most once. 
The right-splitter of a vertex $x$ corresponds to the $y$ we are considering
when we remove $x$ from the set \texttt{PotBeg}.
\end{remark}

\begin{proposition} [Suffix property]

\label{suffix_prop}
The set of the beginnings
of the intervals of $\mathcal{F}$ ending at $y$ forms a suffix of $potBeg(y)$.
More precisely, this set is equal to $potBeg(y) \cap [LSplitter(y)+1,y]$.
\end{proposition}

This last property will be used in particular in the enumeration of $\mathcal{F}$ (see Algorithm \ref{algo_enum}).
Moreover, it will be also useful in the computation of the tree-decomposition of the common intervals
of two permutations (see section~\ref{tree_dec_perm}).

\subsection{Generic algorithms}

We assume that we have exhibited a notion of \textit{potential beginning} and a 
notion \textit{potential end} as defined above. Under this assumption, we here
give very general method to deal with the family $\mathcal{F}$.
In the section \ref{applications} we will explicit how theses approaches can be applied
to specific combinatorial structures.

We first describe a generic algorithm, using two symmetrical sweeps, to compute a representation by generator 
for $\mathcal{F}$. Then we describe another one that enumerates all the elements of $\mathcal{F}$.

\subsubsection{A generic algorithm to compute a generator}

We give here a very simple algorithm to compute a generator representing $\mathcal{F}$.
Our algorithms proceed in two sweeps to compute the couple of vectors $(LSplitter, RSplitter)$.
Recall that this computation answers the problem since $(LSplitter+1,RSplitter-1)$ is a generator 
of $\mathcal{F}$. 
During the first sweep we consider the vertices in increasing order and we maintain the set of potential 
beginnings (of the current $y$) in order to compute the vector $RSplitter$. 
The second sweep is symmetrical: we consider the vertices in decreasing order and maintain the set
of potential ends in order to compute the vector $LSplitter$.

\vspace{1em}
\begin{function}[H]
   \DontPrintSemicolon
   \caption{Computation of the generator()}
   ComputeRightSplitter() \;
   ComputeLeftSplitter() \;
\end{function} 
\vspace{1em}

Using these 2 vectors one can check in constant time if a given interval is in $\mathcal{F}$ using the following function:

\vspace{1em}
\begin{function}[H]
\caption{isInFamily(x,y)}
   Return ($x > LSplitter[y]$ and $y < RSplitter[x]$)
\end{function}
\vspace{1em}

For simplicity, we will only explicit the computation of $RSplitter$ ($LSplitter$ being obtained symmetrically).
To do this computation, we consider the vertices $y$ in increasing order and 
we maintain the set \texttt{PotBeg} of the potential beginnings of the current $y$.

When we have considered the vertices $\{1, \ldots, y-1\}$ and computed the set
$potBeg(y-1)$, the update to $potBeg(y)$ can be done by removing from \texttt{PotBeg} 
the vertices $x$  that are not a potential beginning of $y$. 
Each time we remove a vertex $x$, we set $RSplitter[x]$ to $y$. 
After this sequence of removals, we add $y$ to our set \texttt{PotBeg}.

\vspace{1em}
\begin{algorithm}[H]

\DontPrintSemicolon
\caption{\label{algo_rs}Generic Computation of RSplitter}

RSplitter $\leftarrow$ [$\infty,\ldots,\infty$] \;
PotBeg $\leftarrow$ $Empty Set$\;
\For{y from 1 to n}
{
   \ForEach{$x$ in PotBeg that is not a potential beginning of $y$}
   {
      RSplitter[x] $\leftarrow$ y \;
      Remove $x$ from PotBeg
   }
   Add y to PotBeg
}
\end{algorithm}
\vspace{1em}

\begin{proposition}
Since each vertex is removed at most once,
if we can perform each removal in time $O(f(n))$ then the whole algorithm is in $O(n f(n))$.
\end{proposition}

In all the specific cases studied in this paper, the set \texttt{PotBeg}
has a simple behaviour 
(for example, it  behaves like a stack in the case of the connected intervals
of a tree) and we can perform each removal in time $O(1)$.

\subsubsection{A generic algorithm to enumerate the intervals}

According to the \textit{Suffix property} (proposition \ref{suffix_prop}), we can easily enumerate
for each $y$ the beginnings of the intervals of $\mathcal{F}$ ending at $y$. To do this,
we consider the elements of \texttt{PotBeg} from right to left, until we find one that is not
such a beginning. We will therefore assume that we have a primitive function \texttt{PotBeg.left} 
that permits to go from one element of \texttt{PotBeg} to the one directly on its left. 
Remark that, whatever the data-structure we use for \texttt{PotBeg}, it is always possible
to use a supplementary doubly-linked list \texttt{L} with an external array indexed,
from $1$ to $n$, that indicates for each $x \in \{1,\ldots,n\}$ the corresponding node
(when it exists) in \texttt{L}. This additional structure permits to compute \texttt{PotBeg.left} 
in constant time. It should be noticed that in all the cases studied in this paper, we do not need to do this,
since the set \texttt{PotBeg} has a very simple behaviour in all these examples.
In this algorithm, we assume that  the while  condition can be checked easily ($O(1)$).

\vspace{1em}
\begin{algorithm}[H]
\caption{\label{algo_enum}Enumeration of the intervals}
PotBeg $\leftarrow$ $Empty Set$\;
\For{y from 1 to n}
{
   \ForEach{$x$ in PotBeg that is not a potential beginning of $y$}
   {
      Remove $x$ from PotBeg
   }
   Add y to PotBeg \;
	Let $x$ be the right-most vertex of PotBeg \;
	\While{$[x,y]$ is an interval of $\mathcal{F}$}
	{
		Output($[x,y]$) \;
		$x$ $\leftarrow$ PotBeg.left(x) \;
	}
}
\end{algorithm}
\vspace{1em}

%\subsubsection{Complexity issues}

Under the assumption that we have exhibited a notion of \textit{potential beginning}
and a notion of \textit{potential end} satisfying the above conditions, 
we can applied the two algorithms. To achieve a good complexity, we have to do
efficiently the updates of the data structures. 

Notice that, to test the while condition in Algorithm \ref{algo_enum}, we can first precompute the vector $LSplitter$
and then test in constant time the condition $ x > LSplitter[y] $
(indeed, since $x$ is in $potBeg(y)$, we only have to test if $y$ is
in $potEnd(x)$). Consequently:

\begin{proposition}
If we have a $O(f(n))$ implementation of Algorithm \ref{algo_rs},
we can derive from it a $O(f(n) + |\mathcal{F}|)$ algorithm to enumerate all the elements of $\mathcal{F}$.
\end{proposition}

\begin{corollary}
If we can perform the removals from \texttt{PotBeg} and \texttt{PotEnd} in time $O(1)$, 
then we obtain optimal algorithms both for the computation of a generator and for 
the enumeration of $\mathcal{F}$.
\end{corollary}

\section{\label{applications}Applications}

In this section, we exhibit a non exhaustive list of applications of the framework 
introduced above. Recall that these methods can be applied when
 considering a family of intervals closed under intersection
for which we have defined a \textit{couple of potentiality}.
To apply the generic algorithms, we just have to specify how to 
update the data structure \texttt{PotBeg} (resp. \texttt{PotEnd}).

In particular we introduce the first algorithms to 
compute the intervals corresponding to subsets of nodes
in a tree $T$ for the two properties: ``being connected''
and ``being a path''.

As we will see, in all these applications, 
\texttt{PotBeg} has a very simple behaviour and
all the obtained algorithms are optimal.\\

\subsection{\label{citree}Connected Intervals of a Tree}

Let $T$ be a tree on vertex set $V = \{1, \ldots, n\}$. We denote $T_{\geq x}$ 
the subgraph of $T$ induced by $\{x, x+1, \ldots, n\}$, $T_{\leq y}$ the subgraph induced by $\{1,2,\ldots,y\}$,
and $T[x,y]$ the one induced by $[x,y] = \{x,x+1,\ldots,y\}$. We say that $[x,y] \subset V$ is a \textit{connected interval}
when $T[x,y]$ is connected and we denote $\mathcal{I}$ the set of connected intervals of $T$.

The problem of finding the connected intervals of a tree is both a generalization of 
the one of finding the common intervals of two permutations and a special case of the "Common Intervals of Tree" problem.
\cite{interval_tree} defined the common intervals of two trees $T_1$ and $T_2$ on vertex set $V$ 
as the subsets of $X \subset V$ such that both $T_1[X]$ and $T_2[X]$ are connected. 
%Remark that in the present paper we only use the word interval to refer to integer intervals
%and then differ from the terminology of \cite{interval_tree}.

Although there exists linear-time algorithms to compute the common intervals of two permutations, 
only $O(n^2)$ algorithms are known to compute a satisfactory representation of the common connected
subsets of two trees (notice that there could be an exponential number of such subsets).
When one of the tree is a path, the problem becomes equivalent to find the connected intervals of a tree 
(by renumbering the vertices in the order of the path).

We here address this special case where one of the tree is a path.
First remark that $\mathcal{I}$ is a family of intervals closed under
intersection. Using the general methods described above, we
give the first $O(n)$ algorithm that computes a convenient 
representation for $\mathcal{I}$ and the first $O(n + |\mathcal{I}|)$ 
algorithm that outputs all the intervals of $\mathcal{I}$. 

%%\subsubsection{Notion of potential beginnings for $\mathcal{I}$}

We can introduce for this problem a very simple notion of \textit{potential beginning} 
(resp. \textit{potential end}).

\begin{definition}[Potential beginning]
For a given $y \in V$, we define the potential beginnings for the end $y$ as the $x \leq y$ such that,
in $T_{\geq x}$, $x$ accesses all the elements of $\{x,\ldots,y\}$.
The potential ends of a vertex $x$ are defined symmetrically.
\end{definition}

\begin{theorem} [Characterization of the connected intervals]
$[x,y] \subset V$ is a connected interval iff $x \in potBeg(y)$ and $y \in potEnd(x)$.
\end{theorem}
\begin{proof}
If $[x,y]\in \mathcal{I}$ it is clear from the previous definition that $x \in potBeg(y)$ 
and $y \in potEnd(x)$. Reciprocally assume that $x \in potBeg(y)$ and $y \in potEnd(x)$.
First remark that the path $P$ between $x$ and $y$ has all its values in $[x,y]$ since it 
is both in $T_{\geq x}$ and $T_{\leq y}$. Moreover for all $z \in [x,y]$, the path joining 
$z$ to $P$ is the intersection of the path between $z$ and $x$ and the one between $z$ and $y$ 
therefore it has all its values in $[x,y]$. From this we derive that $[x,y]$ is a connected interval.
\end{proof}

The corresponding notions of right-splitter and left-splitter are then:

\begin{proposition} [Splitters]
For a given $x$, the right-splitter of $x$ is the minimum
$z>x$ such that there exists a vertex $x' < x$ on the path between $x$ and $z$. 
The notion of left-splitter is symmetrical.
\end{proposition}  

In this special case, \texttt{PotBeg} has a simple behaviour.

\begin{proposition} [Stack behaviour of \texttt{PotBeg}]
In this context, \texttt{PotBeg} behaves like a stack (when we consider the $y$ in increasing order).

More formally: If $x < x'$ are in $potBeg(y-1)$ then , if $x'$ is still in $potBeg(y)$ then so is $x$.
\end{proposition}
\begin{proof}
If $x'$ is still in $potBeg(y)$ then, since $x$ accesses $x'$ in $T_{\geq x}$ and 
$x'$ accesses $y$ in $T_{\geq x'}$, so a fortiori in $T_{\geq x}$, we have that $x$
accesses $y$ in $T_{\geq x}$ . We conclude that $x$ is a potential beginning for $y$.  
\end{proof}

This behaviour of \texttt{PotBeg} will be really useful in our algorithm since, as we will see, 
it is easy to maintain this stack. $RSplitter[x]$ is then simply obtained as the 
$y$ we are considering when $x$ is popped from the stack. 

If $x$ is a potential beginning of $y-1$, we can test if it is still one for $y$ 
by checking that $x$ accesses $y$ in $T_{\geq x}$. In other words, we only have 
to check that the path between $x$ and $y$ does not contain values less than $x$.

To do this test in constant time, we will ask for the minimum on the path between $x$ and $y$ to check
if it is less than $x$ or not.
We can get this minimum by computing the lowest common ancestor (LCA) of $x$ and $y$ in the \textit{Cartesian Tree} 
of $T$. In the RAM model, this Cartesian Tree can be computed in linear time and the LCA queries can
be answered in constant time with a linear precomputation. For further details on these data structures, 
please refer to \cite{cart_tree}. We thus assume that we have a function $MinOnPath(x,y)$ that outputs in time $O(1)$ the minimum 
on the path between $x$ and $y$.

From these remarks, we obtained the following $O(n)$ specialization of Algorithm \ref{algo_rs}:

\begin{algorithm}[H]

\caption{\label{algoci}ComputeRightSplitter}
\label{algoRB}
\DontPrintSemicolon

RSplitter $\leftarrow$ [$\infty,\ldots,\infty$] \;
PotBeg.push(1) \;
\For{y from 2 to n}
{
   \While{MinOnPath(PotBeg.top(),y)) <  PotBeg.top()}
   {
      RSplitter[PotBeg.top()] $\leftarrow$ y \;
      PotBeg.pop()
   }
   PotBeg.push(y)
}
\end{algorithm}

Since we can test the pop condition in constant time,
we can also implement Algorithm \ref{algo_enum} to run in time $O(n + |\mathcal{I}|)$.

\subsection{Paths in a Tree}

We have previously seen how to compute in linear time the intervals 
$I$ such that $T[I]$ is connected. We here give
a linear time algorithm to compute the intervals $I$ such that $T[I]$ 
is a path. We denote $\mathcal{P}$ this latter family.

The notion of \textit{potential beginnings} we will use is:

\begin{definition}
$x$ is a potential beginning of $y$ when both the following conditions are satisfied:
\begin{enumerate}
\item $x$ accesses all vertices of $[x,y]$ in $T_{\geq x}$ 
\item $T[x,y]$ is contained in a path of $T$.
\end{enumerate}
\end{definition}

We describe now the data structure we use to update efficiently \texttt{PotBeg}.
When we have computed $potBeg(y-1)$, we have to remove from \texttt{PotBeg} all
the $x$ that are not a potential beginning for $y$.
We proceed in two steps:
\begin{enumerate}
\item We remove from \texttt{PotBeg} all the $x$ that do not satisfy the former condition
\item Then we remove the remaining ones that do not satisfy the latter one.
\end{enumerate}

Remark that since the first condition is identical to the one of section \ref{citree},
we can perform the first step exactly as in Algorithm \ref{algoci} therefore
we will only describe the second step. The second condition will be called the 
\textit{alignment condition}.\\

Since \texttt{PotBeg} is contained in a path, we can
consider  the doubly-linked list \texttt{L} that contains the vertices of \texttt{PotBeg}
in the order in which they appear in the path. During the first step, we will
just do some removals from \texttt{L}, each vertex being removed in constant time.

\begin{proposition}
The set of vertices removed from \texttt{PotBeg} during the second step 
form a prefix of \texttt{PotBeg}.
\begin{proof}
If $[x,y]$ is not contained in a path then
neither is $[x',y]$ for $x'<x$.
\end{proof}
\end{proposition}

During this second step, we will remove the minimum element of \texttt{PotBeg}
while the alignment condition is not respected. To check this, 
we can test if the set of three vertices composed of the
two ends of $L$ and $y$ are aligned or not (this can be done in
constant time using a fixed number of LCA computations).
The insertion of $y$ into \texttt{L} can also be done easily.

We can thus perform each removal in constant time. Therefore
we can build a generator for $\mathcal{P}$ in linear time
and also enumerate $\mathcal{P}$ in optimal time.

\subsection{Closed intervals of a DAG}

Let $G$ be a directed acyclic graph (DAG) with vertex set $V(G) = \{1,\ldots,n\}$
and $m$ arcs.
A \textit{closed interval} is an integer interval $[x,y] \subset V(G)$ such that
all accessible vertices from a $z \in [x,y]$ are in $[x,y]$.
It is easy to prove that this family of intervals is closed under intersection.
Given an interval $[x,y] \subset V(G)$, we denote $Cl([x,y])$ the set
of vertices reachable from at least one of the vertices of $[x,y]$.

\begin{figure}[h]
\begin{center}
\includegraphics[scale=0.8]{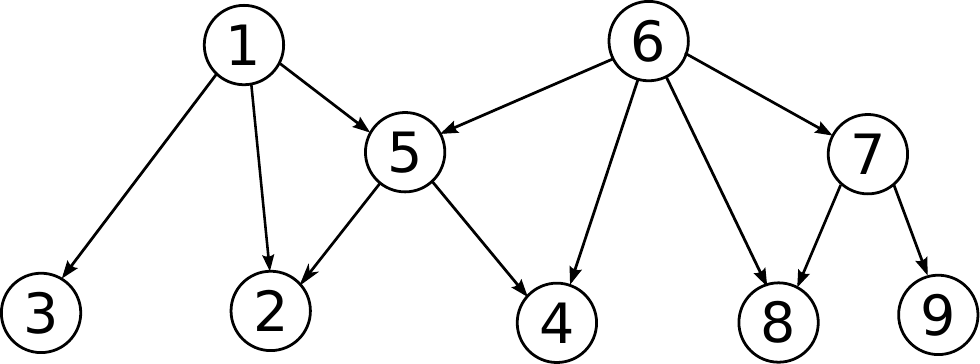}
\caption{An example of DAG on vertex set $\{1,\ldots,9\}$. 
$[2,5]$ is a closed interval whereas $[4,9]$ is not.}
\end{center}
\end{figure}

In this context, it is natural to define the notion of potential beginning
as follows:

\begin{definition} [Potential beginning]
For a given $y$, we say that $x \in \{1, \ldots, y\}$ is a potential.
beginning of $y$ when $x \leq \min Cl([x,y])$. (The notion
of potential end is defined symmetrically).
\end{definition}

As in the case of the connected intervals of a tree, we have the following properties:

\begin{lemma}
With this notion of potential beginning: 
\begin{enumerate}
\item \texttt{PotBeg} behaves like a stack when we consider
the $y$ in increasing order
\item $[x,y]$ is a closed interval iff $x \in potBeg(y)$ and $y \in potEnd(x)$.
\end{enumerate}

\begin{proof}
We first show that \texttt{PotBeg} behaves like a stack. 
Consider $x < x'$ two potential beginnings of $y-1$.
If $x$ is not a beginning for $y$ then there exists $z < x$
reachable from $y$. Since $z < x'$, $x'$ is not a potential beginning for $y$.

We now show the second property.
If $[x,y]$ is a closed interval then $Cl([x,y]) = [x,y]$, therefore
$x$ is a potential beginning of $y$ and $y$ is a potential end of $x$. Reciprocally,
if $x$ is a potential beginning of $y$ and $y$ is a potential end of $x$ then for all 
$z \in Cl[x,y]$, we have that $x \leq z \leq y$ thus $Cl([x,y]) = [x,y]$ and $[x,y]$
is a closed interval.

\end{proof}
\end{lemma}

Les us show how to compute the vector $RSplitter$. 
The only difference with the connected intervals of a tree is the pop condition.
In this context, it is easy to check if the top $x$ of the stack is
still a potential beginning for $y$ since we just have to test
if the minimum label reachable from $y$, denoted $MinBelow(y)$ is 
less than $x$ or not.
To perform this test, we first pre-compute, using dynamic programming,
the vector $MinBelow$ in time $O(n+m)$.

\begin{algorithm}[H]
\caption{Computation of RSplitter for Closed Interval of a DAG}
\DontPrintSemicolon

RSplitter $\leftarrow$ [$\infty,\ldots,\infty$] \;
PotBeg.push(1) \;
\For{y from 2 to n}
{
   \While{MinBelow(y) <  PotBeg.top()}
   {
      RSplitter[PotBeg.top()] $\leftarrow$ y \;
      PotBeg.pop()
   }
   PotBeg.push(y)
}
\end{algorithm}

We thus have a linear-time algorithm that computes a generator 
representing the family of closed intervals of a DAG (and also a linear time
algorithm to enumerate all these intervals).

\subsection{Other examples}

Let us present here a non exhaustive list  of problems which can be solved uisng the above  framework. 
In  the following,  $P$ will always denote a permutation,
$T$ a tree, and $D$ a $DAG$ (Direct Acyclic Graph):

\begin{itemize}
   \item \textbf{A}: intervals $I$ such that $P(I)$ is an interval
   \item \textbf{B}: intervals $I$ such that $P(I) = I$
   \item \textbf{C}: intervals $[x,y]$ such that $P([x,y]) \subset [P(x),P(y)]$
   \item \textbf{D}: intervals $[x,y]$ such that $P([x,y])=[P(x),P(y)]$
   \item \textbf{E}: intervals $I$ such that $T[I]$ is connected.
   \item \textbf{F}: intervals $I$ such that $T[I]$ is contained in a path
   \item \textbf{G}: intervals $I$ such that $T[I]$ is a path  
   \item \textbf{H}: intervals $I$ such that $D[I]$ is closed  
\end{itemize}

All of them describe classes that can be represented by a generator.
The class \textbf{D} has played a central role in comparative genomic.
\cite{cabbage_turnip} introduced it to compare the genomes of
cabbage and turnip, the intervals of this class are called \textit{hurdles}.

\vspace{1em}

\begin{tabular}{|c|c|c|}
   \hline
   \textbf{Class} & \textbf{Potential-beginning of $y$} & \textbf{Behaviour of \texttt{PotBeg}} \\ 
   \hline
   A & see Definition \ref{def_pot_beg_path} & Stack \\% f(I) inter
   \hline
   B & $x$ s.t. $\min P([x,y]) \geq x$ & Stack \\% f(I) = I
   \hline 
   C & $x$ s.t. $\min P[x,y] \geq P(x)$ & Stack \\% f([x,y]) incl dans [f(x),f(y)] 
   \hline
   D & intersection of A and C & Stack \\% f([x,y]) egale [f(x),f(y)] 
   \hline
   E & $x$ s.t. no $x' < x < z \leq y$ with $x'$ between $x$ and $z$ & Stack \\
                % T[I] connected
   \hline
   F & $x$ s.t. $[x,y]$ aligned in $T$ & Queue \\% T[I] dans un chemin
   \hline
   G & intersection of E and F & Queue-Stack \\% T[I] un chemin
   \hline
   H & $x$ s.t. $\min Cl([x,y]) \geq x$ & Stack \\% D[I] clos
   \hline
\end{tabular}

\vspace{1em}

In all the previous cases, we have obtained 
optimal algorithms both for the computation of a generator
and the enumeration. 
Furthermore 
for the classes \textbf{B} and \textbf{D}, all
the algorithms of the section \ref{tree_dec_perm}
can be adapted to these families. In particular,
we can compute their decomposition tree
in linear time.

\section{\label{tree_dec_perm}Tree Decomposition of common intervals of permutations}

In this section, we address the problem of computing the connected intervals of a path.
This problem is equivalent to the problem of finding the common intervals of two permutations:
given two permutations $P_1, P_2$ of $\{1,\ldots,n\}$, compute the subsets of $\{1,\ldots,n\}$
that appear consecutively both in $P_1$ and $P_2$. 
First remark that we can assume that $P_2 = Id_n = (1,2,\ldots,n)$ 
(by renumbering the elements). 
The problem then becomes: 
given a permutation $P$ of $\{1,\ldots,n\}$, compute the integer intervals $I$ such that 
$P(I)$ is an integer interval. These integer intervals are exactly the connected intervals
of a path over $\{1,\ldots,n\}$ that visits the vertices in the order given by $P$.

In the literature, the prevalent formulation of the problem is the computation
of the common intervals of two permutations.
It appears especially in comparative genomic: 
if the genomes of two species are close, then we expect that important parts 
coincides up to some reordering of the genes. 
It also models the notion of gene cluster: 
several genes that present functional associations are expected to appear 
consecutively.

This problem of finding the common intervals of two permutations 
was introduced by \cite{uno_yag} in 2000. 
They propose an optimal, but complex,
algorithm that enumerates the $K$ common intervals in time $O(n + K)$.
\cite{heb_stoye} introduced the notion of irreducible intervals and obtained an 
$O(kn+K)$ algorithm that outputs all $K$ common intervals a $k$ permutations.
\cite{revisit_uno} introduced the tree structure of this family of 
intervals and presented a linear time algorithm to compute this tree.
 \cite{BCMR08} presented a 
simplest linear time algorithm,
introducing the notion of generator that we use to represent the connected intervals of a tree.

In all the section, we consider the permutation $P$ given by the order in which the
path visits the vertices (since this order is defined up to a reversal, we arbitrarily 
choose one of the two possible directions). The connected intervals $\mathcal{I}$ of our path 
are exactly the intervals $I$ such that $P(I)$ is an interval. 
Representing the permutation $P$ in two dimensions (\ref{square}), 
these connected intervals are represented by "squares."

\begin{figure}[h]
   \begin{center}
      \includegraphics[scale=0.6]{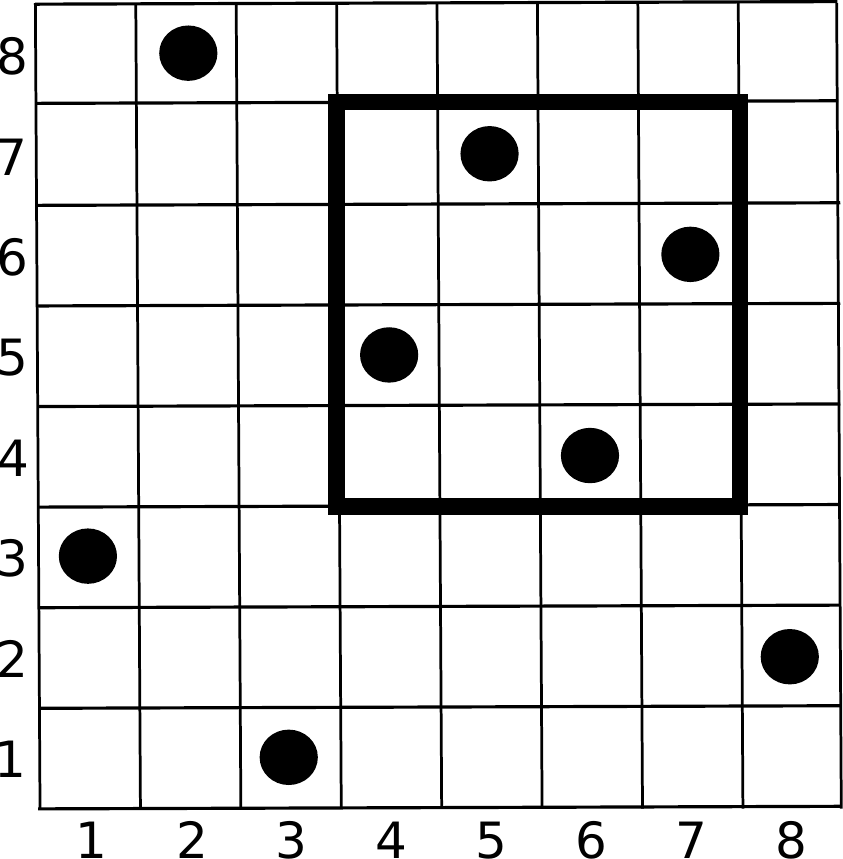}
   \end{center}
   \caption{\label{square}A representation in two dimensions of the permutation $(3,8,1,5,7,4,6,2)$. 
The connected interval $[4,7]$ is represented with a square.}
\end{figure}

In the tree case, $\mathcal{I}$ is closed under union and intersection of its overlapping members.
When $G$ is a path, $\mathcal{I}$ is also closed under the difference of its overlapping members.
A family closed under union, intersection and difference of its overlapping members is called
\textit{weakly partitive} and admits a canonical tree of decomposition \cite{partitive_hyper}.
Since $\mathcal{I}$ is moreover a family of intervals, this tree has a particularly simple structure 
that we will explicit. 

In the previous section, we presented a linear time algorithm to compute a generator 
representing the family $\mathcal{I}$ when $G$ is a tree. We can use this algorithm 
to compute a generator for 
the case of a path and then use the procedure described in \cite{common_k_perm} to create from a generator 
the tree representing a weakly partitive family of intervals.
However, using the specific properties of this special case, we can obtain a simplest 
linear time algorithm that only uses basic data 
structures like stacks and directly builds the tree.
We will use, like in the tree case, the notion of \textit{potential beginning} and \textit{potential end}.

We will first address the problem of checking the simplicity of a permutation. 
A permutation is called \textit{simple} (or \textit{prime}) when all its corresponding 
connected intervals are trivial (i.e.\ of length 1 or $n$).
Simple permutations are of first interests in the combinatorial study of permutation classes. 
In particular, \cite{simple_perm} if the simplicity of a permutation could 
be checked in linear time. Of course, the computation of the tree decomposition of $\mathcal{I}$ 
answers this question. However, we present a very simple linear time algorithm 
to answer it. We afterwards extend this algorithm to build the decomposition tree of $\mathcal{I}$.

\subsection{A very simple algorithm to test the simplicity of a permutation}

Given the permutation $P$, we present a very simple linear time 
algorithm that computes a non trivial connected interval
when there exists one. 
This algorithms will cover the main ideas we will use
to compute the decomposition tree of $\mathcal{I}$.

We will, as in the tree case, consider the elements one by one and  maintain 
the potential beginnings. For convenience, we redefine the notion 
of potential beginning in this new context (even if it coincides with 
the one given inherited from the tree case). 

\begin{definition} [Potential beginning]
\label{def_pot_beg_path}
Given an end $y$, we say that $x$ is a potential beginning for $y$ when 
the following conditions are both satisfied: 
\begin{itemize}
   \item $\nexists z_1 < x < z_2 \leq y, P(x) < P(z_1) < P(z_2)$
   \item $\nexists z_1 < x < z_2 \leq y, P(x) > P(z_1) > P(z_2)$
\end{itemize}
\end{definition}

Only the elements of $potBeg(y)$ have a chance to be the beginning of a 
connected interval ending after $y$. Indeed, if there exists for example
$z_1 < x <z_2$ with $P(x) < P(z_1) < P(z_2)$, then for $y' \geq y$, 
$P([x,y'])$ contains $P(x)$ and $P(z_2)$ but not $P(z_1)$ that is
between $P(x)$ and $P(z_2)$ so $[x,y'] \not \in \mathcal{I}$. 

For the more general case of a tree, we have shown 
that \texttt{PotBeg} behaves like a stack (when we consider the $y$ 
in increasing order) and that the beginnings of the connected intervals
ending at $y$ form a suffix of this stack. So for a given $y$, 
we will just have to check the head of \texttt{PotBeg} to detect if
there is a non trivial connected interval ending at $y$.

\subsubsection{Maintaining the stack}

To maintain the stack, i.e.; update $potBeg(y-1)$ to $potBeg(y)$,
we just have to pop the top $x$ of the stack while there 
exists $z < x$ such that $P(x) < P(z) < P(y)$ or 
$P(x) > P(z) > P(y)$.

In order to check this condition in constant time, 
we precompute for each $x$ the values 
$minGreaterOnLeft[x] = \min \{P(z) | z < x, P(z)>x\}$
and
$maxSmallerOnLeft[x] = \max \{P(z) | z < x, P(z)<x\}$.
This precomputation is a classic one and can be done easily in linear 
time with a stack. 

Precisely, we have to pop $x$ when:
\begin{itemize}
   \item $minGreaterOnLeft[x] < P(y)$
   \item or $maxSmallerOnLeft[x] > P(y)$
\end{itemize}

We can thus check in constant time (using simplest data structures than in the tree case),
if the top of the stack has to be popped.

\subsubsection{Detection of a non trivial connected interval}

Recall that to detect if their exists a non trivial connected interval ending at $y$, we
just have to check if the greater potential beginning $x<y$ of $potBeg(y)$ 
is the beginning of a connected interval ending at $y$.

Denoting $maxi(x,y) = \max \{P(z) | x \leq z \leq y\}$ and 
$mini(x,y) = \min \{P(z) | x \leq z \leq y\}$, we can check if $[x,y] \in \mathcal{I}$ 
by  testing if $maxi(x,y) - mini(x,y) = y - x$. 
In order to compute $maxi(x,y)$ and $mini(x,y)$ when we want to perform this test, 
we have to maintain the maximum and minimum between each pair of consecutive potential beginnings
in the stack \texttt{PotBeg} (as shown in the algorithm).

\begin{algorithm}[H]
\caption{\label{one_ci}Detection of a non trivial connected interval}
\DontPrintSemicolon

\For{y from 1 to n}
{
   mini $\leftarrow$ P(y) \;
   maxi $\leftarrow$ P(y) \;
   \While{PotBeg.size()>0 and (minGreaterOnLeft[PotBeg.top()] <  P(y) or 
          maxSmallerOnLeft[PotBeg.top()] > P(y))}
   {
      mini $\leftarrow$ min(mini, minBefore.top()) \;
      maxi $\leftarrow$ maxi(maxi, maxBefore.top()) \;
      PotBeg.pop(), minBefore.pop(), maxBefore.pop() \;
   }
   x $\leftarrow$ PotBeg.top() \;
   PotBeg.push(y), minBefore.push(mini), maxBefore.push(maxi) \;
   \If{max(maxi, perm[x]) - min(mini,perm[x]) = y-x}
   {
      Return [x,y] \;
   }
}
\end{algorithm}

\subsection{A very simple algorithm to enumerate the connected intervals}

Notice that, since the beginnings of the connected intervals ending 
at $y$ form a suffix of the stack \textit{PotBeg}, we can enumerate all the
$K$ connected intervals in time $O(n + K)$ by replacing the last if-statement
of Algorithm \ref{one_ci} by: 

\begin{function}[H]
\caption{Enumeration of the beginnings()}
\DontPrintSemicolon
iPotBeg $\leftarrow$ PotBeg.size()-1 \;
x $\leftarrow$ PotBeg[iPotBeg] \;
mini $\leftarrow$ P(y) \;
maxi $\leftarrow$ P(y) \;

\While{iPotBeg $\geq$ 0 and
(max(maxi, perm[x]) - min(mini,perm[x]) = y-x)}
{
   Output($[x,y]$) \;
   mini $\leftarrow$ min(mini, minBefore[iPotBeg]) \;
   maxi $\leftarrow$ maxi(maxi, maxBefore[iPotBeg]) \;
   iPotBeg $\leftarrow$ iPotBeg - 1 \;
   \If{iPotBeg $\geq$ 0}
   {
      x $\leftarrow$ PotBeg[iDeb] \;
   }
}

\end{function}

This algorithm is much simpler than the one given by Uno and Yagiura \cite{uno_yag}

\subsection{The tree representation of the family}

Recall that two sets overlap when they intersect without inclusion.
Each time we consider a family $\mathcal{F}$ on a ground set $V$, we assume that
$\mathcal{F}$ contains $V$ and all the singletons of $V$.
A \textit{weakly partitive family} is a family $\mathcal{F}$ such that if $A,B$ are two members of $\mathcal{F}$ 
that overlap then $A \cup B$, $A \cap B$, $A \setminus B$, $B \setminus A$ are in $\mathcal{F}$.
It is easy to check that, when $G$ is a path, $\mathcal{I}$ is a weakly partitive family of intervals. 
\cite{partitive_hyper} showed that a weakly partitive family admits a 
tree representation (of linear size). 

We will here explicit this tree in the case of the family of the connected intervals of a permutation $P$.
The nodes of the tree are given by the overlap-free members of the family: 
a member is overlap-free when it does not overlap 
any other. The family $\mathcal{L}$ of overlap-free members of $\mathcal{I}$ is laminar by definition and 
then can be represented with a tree $\mathcal{T}_\mathcal{L}$ where the parent of 
$u\in\mathcal{L}$ is the smallest $v\in\mathcal{L}$ that strictly contains $u$.
This tree will be labelled in order to represent the whole family $\mathcal{I}$.

\begin{lemma}
If $X \in \mathcal{I} \setminus V$ then $X$ is an 
union of children of the smallest overlap-free member of $\mathcal{I}$ that contains it.
\end{lemma}

\begin{definition} [Quotient Family]
Let $u$ be a node of $\mathcal{T}_\mathcal{L}$.
We denote $children(u)$ the list of the children of $u$ in $\mathcal{T}_\mathcal{L}$ 
given from left to right. If $v_1$ and $v_2$ are two distinct children of 
$u$, $P(v_1)$ and $P(v_2)$ are two disjoint intervals.
Denoting $P(v_1) \preceq P(v_2)$ when $max(P(v_1)) < min(P(v_2))$, we thus obtain a 
total order. The \textit{quotient} family $\mathcal{Q}(u)$ of $u$ is defined as the 
permutation given by $\preceq$ on $children(u)$.
\end{definition}

\begin{theorem} [Description of quotients]
Let $u$ be an internal node of $\mathcal{T}_\mathcal{L}$ having $k$ children.
Exactly one of the following assertions holds:
\begin{enumerate}
   \item $\mathcal{Q}(u)$ is the increasing permutation of $\{1,\ldots,k\}$.
   \item $\mathcal{Q}(u)$ is the decreasing permutation of $\{1,\ldots,k\}$
   \item $k \geq 3$ and $\mathcal{Q}(u)$ is simple.
\end{enumerate}
\end{theorem}

\begin{figure}
\begin{center}
\begin{minipage}[c]{0.4\textwidth}
\includegraphics[scale=0.4]{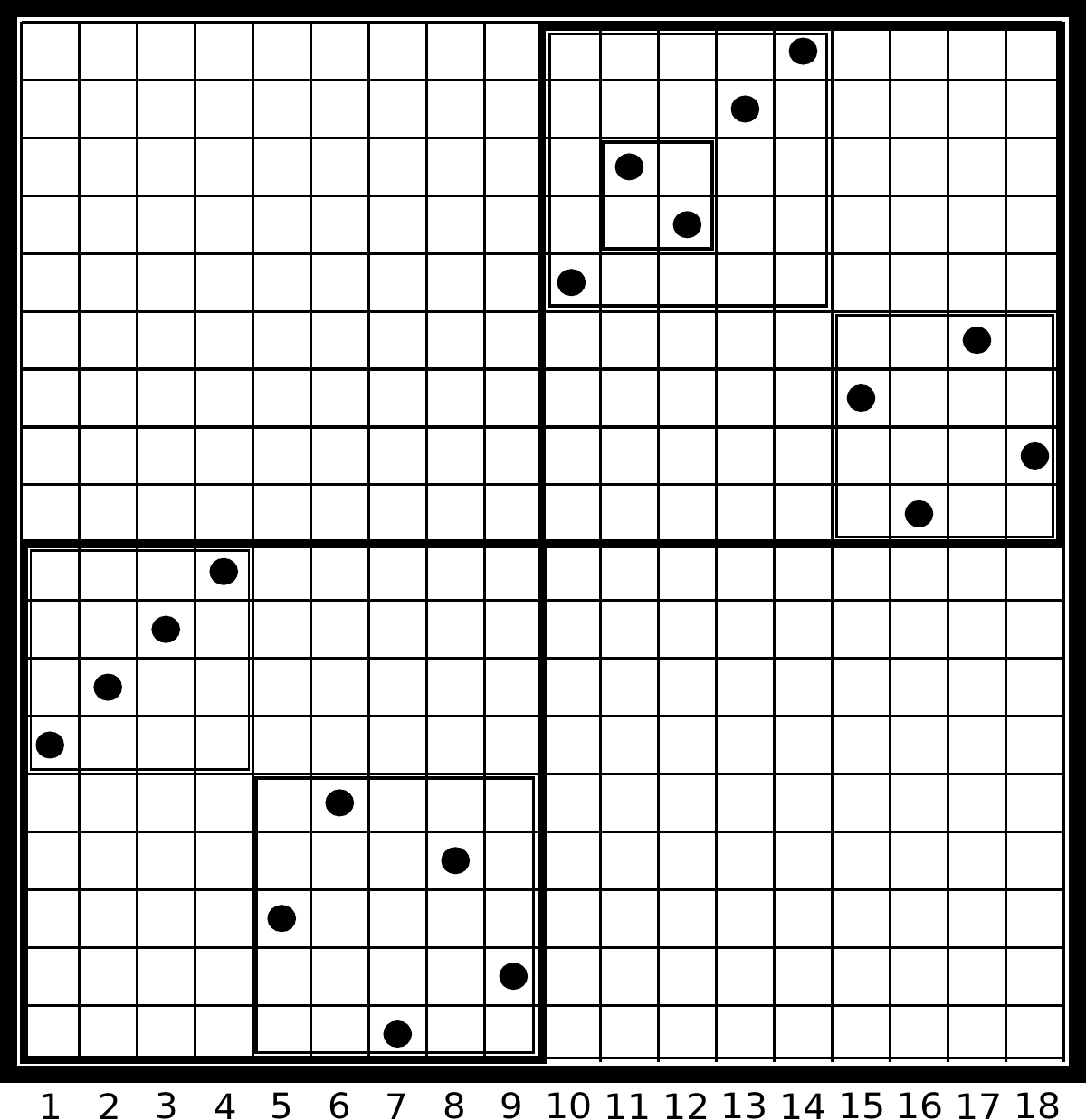}
\end{minipage}
\begin{minipage}[c]{0.4\textwidth}
\includegraphics[scale=0.5]{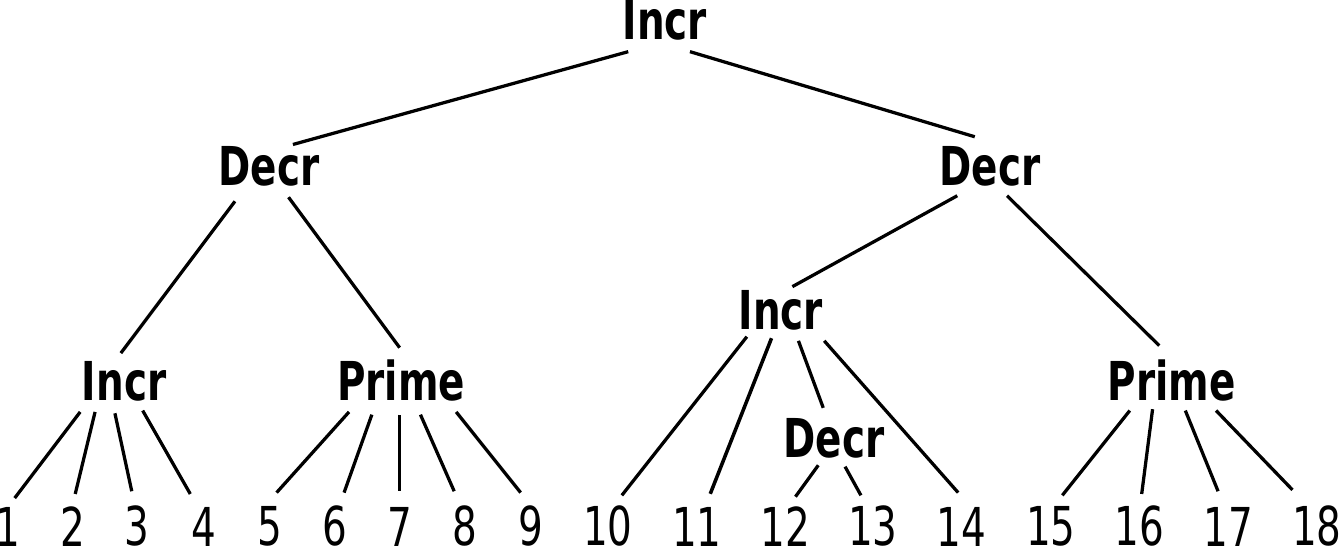}
\end{minipage}
\end{center}
\caption{Tree decomposition of the permutation 
$(6,7,8,9,3,5,1,4,2,14,16,15,17,18,12,10,13,11)$}
\end{figure}

This theorem shows that we can label each internal node of 
$\mathcal{T}_\mathcal{L}$ either as \textit{Prime}, 
\textit{Increasing} or \textit{Decreasing}. From the obtained
labeled tree, we can easily derive the whole family $\mathcal{I}$. 

\subsection{Computation of the tree}

\subsubsection{General description}

Extending the previous algorithm, we present an algorithm
that computes the decomposition tree of $\mathcal{I}$.
Roughly speaking, we can just process as before and contract the connected 
intervals found on the fly.

Each time we contract a connected interval, we will manipulate it
like a singleton. More generally we will speak about \textit{nodes}.
\texttt{PotBeg} is now a stack of nodes and all the variables used 
in the previous algorithms can be adapted to support this notion 
of node.

We have defined above the decomposition tree of a permutation. In fact,
we can define the decomposition forest of any injection from $\{1,\ldots,k\}$ to $\{1,\ldots,n\}$
by considering the overlap-free connected intervals.

The general idea is still to consider the $y$ in increasing order. 
At each step, we maintain the decomposition forest $\mathcal{F}_y$ of the 
restriction of $P$ over $\{1,\ldots,y\}$. We will see how to update
$\mathcal{F}_{y-1}$ to $\mathcal{F}_y$. To to this update, we will have
two operations $tryExtension$ and $tryPrimeCreation$.

\begin{algorithm}[H]
\caption{\label{compTree}Computation of the tree}

\For{y from 1 to n}
{
   S.push(Node({y})) \;
   stable $\leftarrow$ false \;
   \While{not stable}
   {
      stable $\leftarrow$ true \;
      \If{tryExtension() or tryPrimeCreation()}
      {
         stable $\leftarrow$ false \;
      }
   }
}
\end{algorithm}

Remark: the $or$ is lazy so that the priority is given
to the extension case.

\subsubsection{Maintaining the forest of decomposition}

Assume that we have considered $\{1,\ldots,y-1\}$ and computed 
the forest $\mathcal{F}_{y-1}$
of the overlap-free connected intervals in the restriction of $P$
over $\{1,\ldots,y-1\}$. Denoting $A_1, \ldots, A_p$ the trees of 
this forest (numbered from left to right), when we consider $y$ 
we build the node $A$ corresponding to the singleton $y$
and we represent this configuration by the notation $<A_1,\ldots,A_p|A>$.
This configuration will evolve until the update to 
$\mathcal{F}_{y}$ is achieved, $A$ representing the tree
decomposition of the restriction of $P$ on its support.

$S = A_1,\ldots,A_p$ can be seen as a stack ( remark that the stack of 
potential beginnings is a substack of $S$) and we will consider
operations considering $A$ and a suffix of $S$. 
Metaphorically speaking, we can consider that $A$ will 
``eat'' head terms of $S$. More precisely, we consider two
kinds of operations.

\paragraph{Operation 1: Monotonic extension (tryExtension)}
Let us describe the case of an \textbf{increasing extension} 
(the case of a decreasing one being symmetrical). 
Consider $I_1$ the connected interval corresponding to the root of $A_p$
and $I_2$ the one corresponding to the root of $A$. 
When $I_1$ is an increasing node and $\max P(I_1) + 1 = \min P(I_2)$, 
we do the following:
\begin{itemize}
   \item Pop $A_p$ from the stack $S$.
   \item A becomes $Add(A_p,A)$.
\end{itemize}
where $Add(T_1,T_2)$ denotes the operation that
returns the tree obtained by appending $T_2$ at the end of the 
list of children of the root of $T_1$.

\paragraph{Operation 2: Prime super-node creation (tryPrimeCreation)}

This operation can be performed when there exists $0 \leq i \leq p$
such that $A_i,\ldots, A_p, A$ form a prime quotient (if it exists,
such an $i$ is unique and is the top of \texttt{PotBeg}). In this case:

\begin{itemize}
   \item Pop $A_i,\ldots,A_p$ from the stack $S$. 
   \item $A$ becomes the prime node whose children are $(A_i,\ldots,A_p,A)$
\end{itemize}

As long as possible, we use this two operations with priority given
to the first one. Since $p$ always decreases, this process terminates.

\subsubsection{Proof of the correctness}

In the following, when we denote $A_i$ a tree whose leaves
are an interval of integers, $S_i$ will denote this interval
(the support of $A_i$).

The correction of this algorithm mainly comes from 
the following lemma:

\begin{lemma}
\label{lemDec}
Assume that we have a forest $A_1,\ldots,A_m$ such that:
\begin{itemize}
   \item for all $1 \leq i \leq m$, $A_i$ is the decomposition tree of the restriction of $P$ to $S_i$
   \item there are no $1 \leq i <j \leq m$ such that 
$S_i \cup \ldots \cup S_j$ is a connected interval
\end{itemize}
then this forest is the decomposition forest.

\begin{proof}
Assume that the forest $A_1,\ldots,A_m$ satisfies the above conditions.

First, we show that a connected interval does not overlap
any of the supports $S_1, \ldots, S_p$. Assume towards contradiction 
that a connected interval $I$ overlaps $S_i$ and we choose $I$
minimal for inclusion.
Without loss of generality, we assume that $I$ overlaps $S_i$
on the right (i.e.; the beginning of $S_i$ is not in $I$). 
Let $S_i,\ldots,S_j$ be the supports that $I$ intersects.
If $j = i+1$ then we have that $S_i \cup S_j$ is a connected interval, that
contradicts the assumptions. We thus have $j > i+1$.
$I$ does not overlap $S_j$ (if not we would contradict its minimality
by considering $I \setminus S_j$), so $S_j \subset I$.
From this we derive that $I \setminus S_i = S_{i+1} \cup \ldots \cup S_j$
is a connected interval, contradiction.

From this property we have that a connected interval is contained in
one of the supports $S_1,\ldots,S_j$. Consequently, every overlap-free 
connected interval is a node of one of the tree $A_1,\ldots,A_m$, since
each $A_i$ is the decomposition tree of its support. We hence
have the decomposition forest.
\end{proof}
\end{lemma}

\begin{theorem}
If $A_1,\ldots A_m$ are the trees of the decomposition forest of $P$
restricted to $\{1,\ldots,k-1\}$, then $A_1,\ldots, A_p,A$, obtained
with the previous algorithm when adding $k$, is the decomposition
forest of $P$ restricted to $\{1,\ldots,k\}$. 

\begin{proof}
From the previous lemma, if we assume that $A$ is the decomposition
tree on its support then, since $A_1,\ldots,A_p$ are 
the decomposition trees on their support and no prime creation is possible,
we obtain the decomposition forest of $P$ restricted to $\{1,\ldots,k\}$.

We then have to demonstrate that $A$ is the decomposition tree
on its support. Recall that we begin with configuration 
$<A_1,\ldots,A_p|A>$ where $A$ is the tree whose only node is
the singleton $k$. Recall moreover that while we can perform
a monotonic extension or a prime creation, we process it
giving priority to the monotonic extension. $p$ and $A$ thus evolves
until we obtain the final configuration.\\

First, we show that we have the two following invariants:

\begin{itemize}
   \item $(i)$ If $A_p$ is increasing and we can do an
increasing extension, then $A$ is not increasing.
   \item $(ii)$ If $A_p$ is decreasing and we can do a
decreasing extension, then $A$, is not decreasing.
\end{itemize}

We show only $(i)$, ($(ii)$ being symmetrical). Assume that $A_p$
and $A$ are increasing nodes and that we can do an increasing
extension. Let $I$ denotes the rightmost child of the root of $A_p$, and
let $I'$ be the first of the root of $A$. $I \cup I'$ would be a connected 
interval that overlaps the support of $A_p$, it contradicts 
the fact that the root of $A_p$ was overlap-free in the decomposition of the 
restriction of $P$ to $\{1,\ldots,k-1\}$. \\

We now show the following invariant:
$A$ is the tree decomposition on its support. 

$\triangleright$ We demonstrate first that a monotonic extension preserves this
property. As before, we only consider the increasing case. We consider the node 
$A$ just after the increasing extension and let $F_1, \ldots, F_s$
be its children ($F_s$ is so the old value of $A$). 
We now only consider the restriction of $P$ to the support of $A$. Assume 
towards contradiction that
one of the nodes of $A$ is not overlap-free. We obtain that there exists
a connected interval that overlaps $F_s$ but it would contradict the
fact that $F_s$ is not increasing (according to $(i)$). Hence the nodes
of $A$ all correspond to overlap-free connected intervals. Reciprocally,
it is straightforward that all overlap-free connected intervals 
of the restriction of $P$ on the support of $A$ are represented by a node
of $A$.

$\triangleright$ Eventually we demonstrate that a prime super-node creation,
when there is no possible monotonic extension, preserves the invariant too.
Let $F_1,\ldots,F_s$ be the children of the prime node created ($F_s$ is
the old value of $A$). Each $F_i$ is the tree decomposition on its support.
If a connected interval would overlap the support of $F_i$, then it would
also intersect $F_s$ and we could have processed a monotonic extension.
Moreover, there is no connected interval different from the support of $A$
that are an union of several $F_i$. From lemma \ref{lemDec}, we have
have the result.

\end{proof}
\end{theorem}

\subsubsection{Complexity}

The linearity of the algorithm comes from the $O(1)$  detection of 
the \textit{monotonic extensions} and \textit{Prime super-nodes creations}
 using the stack of potential beginnings. Moreover,
these two kinds of operations take a constant time to be performed.
Since each of these two kinds of updates create at least one arc
in the final decomposition tree, there is at most $n-1$ such
updates. Hence the whole complexity is $O(n)$.

section{Conclusion}
The framework presented here not only simplify existing algorithms, but it allows to solve optimally new problems as developed in section 3.
We are convinced that this framework can also be applied to improve some algorithms dealing with permutations avoiding some patterns as defined in \cite{P13}.
In a companion paper \cite{BH13} we have studied  the case
where $G$ is a graph and $\Pi=$ "being connected" and develop a variation of this framework using more sophisiticated data structures.

\bibliographystyle{plain}

\begin{thebibliography}{10}

\bibitem{simple_perm}
M.~H. Albert, M.~D. Atkinson, and M.~Klazar.
\newblock The enumeration of simple permutations.
\newblock {\em J. Integer Seq}, 6, 2003.

\bibitem{BBCR04}
Marie-Pierre B{\'e}al, Anne Bergeron, Sylvie Corteel, and Mathieu Raffinot.
\newblock An algorithmic view of gene teams.
\newblock {\em Theor. Comput. Sci.}, 320(2-3):395--418, 2004.

\bibitem{BH13}
Ismael Belghiti and Michel Habib.
\newblock Connected intervals of a graph.
\newblock Technical report, LIAFA, University Paris Diderot, 2013.

\bibitem{BCMR08}
Anne Bergeron, Cedric Chauve, Fabien de~Montgolfier, and Mathieu Raffinot.
\newblock Computing common intervals of k permutations, with applications to
  modular decomposition of graphs.
\newblock {\em SIAM J. Discrete Math.}, 22(3):1022--1039, 2008.

\bibitem{BMLPV05}
Fr{\'e}d{\'e}ric Boyer, Anne Morgat, Laurent Labarre, Jo{\"e}l Pothier, and
  Alain Viari.
\newblock Syntons, metabolons and interactons: an exact graph-theoretical
  approach for exploring neighbourhood between genomic and functional data.
\newblock {\em Bioinformatics}, 21(23):4209--4215, 2005.

\bibitem{revisit_uno}
Binh-Minh Bui-Xuan, Michel Habib, and Christophe Paul.
\newblock Revisiting \textsc{T. Uno} and \textsc{M. Yagiura's} algorithm.
\newblock In {\em ISAAC}, pages 146--155, 2005.

\bibitem{BXHP08}
Binh-Minh Bui-Xuan, Michel Habib, and Christophe Paul.
\newblock Competitive graph searches.
\newblock {\em Theor. Comput. Sci.}, 393(1-3):72--80, 2008.

\bibitem{partitive_hyper}
M.~Chein, Michel Habib, and M.~C. Maurer.
\newblock Partitive hypergraphs.
\newblock {\em Discrete Mathematics}, 37(1):35--50, 1981.

\bibitem{cabbage_turnip}
Sridhar Hannenhalli and Pavel~A. Pevzner.
\newblock Transforming cabbage into turnip: Polynomial algorithm for sorting
  signed permutations by reversals.
\newblock {\em J. ACM}, 46(1):1--27, 1999.

\bibitem{interval_tree}
Steffen Heber and Carla~D. Savage.
\newblock Common intervals of trees.
\newblock {\em Inf. Process. Lett.}, 93(2):69--74, 2005.

\bibitem{heb_stoye}
Steffen Heber and Jens Stoye.
\newblock Finding all common intervals of k permutations.
\newblock In {\em CPM}, pages 207--218, 2001.

\bibitem{P13}
Adeline Pierrot.
\newblock {\em Combinatoire et algorithmique dans les classes de permutations}.
\newblock PhD thesis, University Paris Diderot, 2013.

\bibitem{uno_yag}
Takeaki Uno and Mutsunori Yagiura.
\newblock Fast algorithms to enumerate all common intervals of two
  permutations.
\newblock {\em Algorithmica}, 26(2):290--309, 2000.

\bibitem{cart_tree}
Jean Vuillemin.
\newblock A unifying look at data structures.
\newblock {\em Commun. ACM}, 23(4):229--239, 1980.

\end{thebibliography}

\end{document}